\documentclass[journal,twocolumn]{IEEEtran}

\usepackage{amsmath,bm}
\usepackage{bbm}
\usepackage{url}
\usepackage{cite}
\usepackage{tabularx}
\usepackage{colortbl,booktabs}
\usepackage{multirow}
\usepackage{threeparttable}
\usepackage{algorithm}
\usepackage{verbatim}

\usepackage{xcolor}
\usepackage{diagbox}
\usepackage{enumerate}
\usepackage{threeparttable}
\usepackage{graphicx}
\usepackage{subfigure}
\usepackage{algpseudocode}
\usepackage{makecell}
\usepackage{nicematrix}
\usepackage{lipsum}
\usepackage{amsthm,amsmath,amssymb}
\usepackage{tablefootnote}

\usepackage{hyperref}
\usepackage{siunitx}
\usepackage{setspace}
\newtheorem{prop}{Proposition}
\newtheorem{remark}{Remark}
\newtheorem{assumption}{Assumption}
\usepackage{optidef}

\title{Toward Value-oriented Renewable Energy Forecasting: An Iterative Learning Approach}

\author{Yufan Zhang, Mengshuo Jia, Honglin Wen~\IEEEmembership{Member,~IEEE,} Yuexin Bian, Yuanyuan Shi~\IEEEmembership{Member,~IEEE.}
\thanks{Yufan Zhang, Yuexin Bian, and Yuanyuan Shi are with the Department of Electrical and Computer Engineering, University of California San Diego. Their work are partially supported by the CA Climate Action, Resilience, and Environmental Sustainability (CARES) grant from UC San Diego. 
Mengshuo Jia is with the Power Systems Laboratory, ETH Z\"urich, 8092 Z\"urich, Switzerland (jia@eeh.ee.ethz.ch). Honglin Wen is with the Department of Electrical Engineering, Shanghai Jiao Tong University. 
Co-corresponding authors: Honglin Wen (linlin00@sjtu.edu.cn) and Yuanyuan Shi (yyshi@ucsd.edu).}
\vspace{-2em}}

\begin{document}

\maketitle
\thispagestyle{empty}
\pagestyle{plain}







\begin{abstract}
Energy forecasting is an essential task in power system operations. Operators usually issue forecasts and use them to schedule energy dispatch in advance. However, forecasting models are typically developed in a way that overlooks the decision value of forecasts.
To bridge the gap, we design a \emph{value-oriented} point forecasting approach for sequential energy dispatch problems with renewable energy sources. At the training phase, we align the training objective with the decision value, i.e., minimizing the overall operating cost. The forecasting model parameter estimation is formulated as a bilevel program. 
Under mild assumptions, we convert the upper-level objective into an equivalent form using the dual solutions obtained from the lower-level operation problems.
In addition, a novel iterative solution strategy is proposed for the newly formulated bilevel program. Under such an iterative scheme, we show that the upper-level objective is locally linear with respect to the forecasting model output and can act as the loss function. Numerical experiments demonstrate that, compared to commonly used forecasts predicting expected realization, forecasts obtained by the proposed approach result in lower operating costs. 
Meanwhile, the proposed approach achieves performance comparable to that of two-stage stochastic programs, but is more computationally efficient. 


\end{abstract}




\section{Introduction}
With the ongoing transition towards net zero emissions in the energy sector, the penetration of renewable energy sources (RESs) has surged over recent decades. Due to inherent uncertainty, RESs cannot be scheduled at will, which has therefore motivated research on renewable energy forecasting to communicate information to the operation of power systems in advance \cite{morales2013integrating}. 
Typically, forecasts are communicated in the form of point forecasts~\cite{wang2018combining}, prediction intervals \cite{zhang2022optimal}, and probability distributions\cite{wen2022continuous,chen2018model}, which offer a snapshot or a whole picture of the predictive distribution of renewable generation. Today, the forecasting community is increasingly relying on data-driven approaches, which estimate parameters of forecasting models based on historical data and some loss functions that measure the quality of forecasts in these data; see the comprehensive review by \cite{gneiting2014probabilistic}.

Usually, forecasts serve as the input to the downstream operational problems and therefore influence decisions. The resulting decision value of the downstream operational problem can also be regarded as the value of forecasts, as suggested by \cite{murphy1993good}. Several previous works have shown that forecasts with higher statistical accuracy may not necessarily lead to desirable value in the operational problems \cite{carriere2019integrated,zhao2021cost,chen2021feature,stratigakos2022prescriptive,zhang2022contextual,zhang2022cost}. Therefore, \textbf{value-oriented forecasting} has been advocated \cite{bertsimas2020predictive,elmachtoub2022smart}, which aims at developing forecasting products that lead to higher value in subsequent decision-making, instead of better statistical accuracy.  

However, training forecasting models to achieve higher operating value presents significant challenges. Different lines of research in the literature address this issue in various ways; see an up-to-date review \cite{mandi2023decision}. The first strand derives the gradient of optimal decision solutions w.r.t. the forecast for training forecasting models via gradient descent \cite{donti2017task,wahdany2023more}, which involves a large-size matrix inverse and results in a computational burden. For operational problems where such a gradient is discontinuous or even vanishes everywhere, such as linear programs where coefficients in the objective function need to be estimated, this approach is not suitable. This motivates the second thread, which designs loss functions for training. ``Smart Predict, then Optimize" (SPO) loss \cite{elmachtoub2022smart} is one of the pioneering works in this area, which relies on a hand-crafted regret loss function. To automate the process of designing loss functions, several studies propose learning the loss function from the data \cite{shah2022decision,zhang2022cost}.

In general, value-oriented forecasting approaches can communicate information that aligns with the operational value. 
Most power system operational problems are organized sequentially \cite{kirschen2018fundamentals}, which includes two deterministic problems made at different times, 
namely in the day-ahead and real-time stages (both referring to the time at which decisions are made relative to the delivery of energy). At the day-ahead stage, a single RES forecast is used as input to the operational problem to settle the quantity of energy to be dispatched the next day. The real-time stage occurs close to the actual delivery of electricity, following the realization of RES. During this stage, corrective decisions are made by solving an operational problem to balance any deviations from the initial schedules. In this setting, the decision value is commonly defined as minimizing the overall operating cost or maximizing the overall welfare/profit at both stages.

One of the most successful applications of aligning the forecasting model with such decision value definition is to settle the offering strategy of wind power producers, where the objective is to maximize the expected overall operating profit (by anticipating the balancing cost at the real-time stage). 
Refs \cite{bitar2012bringing,pinson2013wind,morales2013integrating} demonstrated that the wind power offering aligns with the Newsvendor problem \cite{arrow1951optimal}, and the optimal offering strategy is to predict the quantile.
The link between the decision value, i.e., the overall profit, and the forecast can be explicitly formulated, and the pinball loss can be used as a surrogate loss function to develop forecasting models.

For complex operational problems involving several decision variables and constraints, a closed-form relationship between the decision value, such as total operating cost, and the forecasts may not be explicit. 
One possible solution is to optimize the forecasting model parameter concurrently with the decision variables. For that, a bilevel program is formulated for the parameter estimation \cite{morales2023prescribing,garcia2021application}. 
In this context, the lower-level optimization encompasses day-ahead operations with the forecast acting as a parameter, while the upper-level problem focuses on optimizing the model parameters along with real-time decisions. To ensure the tractability of such a program, forecasting models are typically simplified to linear regression models, which might restrict their effectiveness.

In this paper, we develop a value-oriented forecasting approach in the context of day-ahead and real-time energy dispatch problems for virtual power plant (VPP) operators.
At the training phase, a bilevel program is formulated for the estimation of the forecasting model parameter, where the upper level optimizes the forecasting model parameter toward the decision value, i.e., minimizing the expected sum of day-ahead and real-time operating costs. To explicitly show the relationship between the decision value and the forecast, we convert the upper-level objective to its dual form, assuming that the energy dispatch problems are linear programs. This conversion allows us to obtain the ``value-oriented'' loss function for training the forecast models. By explicitly deriving the loss function, sophisticated regression models, such as Neural Networks (NNs), can be employed as forecasting models. This approach surpasses the limitations of earlier studies \cite{garcia2021application,morales2023prescribing} that were confined to linear models. Also, the explicit loss function allows us to directly obtain the gradient of the decision value w.r.t. the forecast. These explicit gradients correspond to the optimal dual solutions for day-ahead and real-time energy dispatch, which can be acquired efficiently and avoids the computational burden of calculating the gradient of optimal solutions w.r.t. the forecast \cite{donti2017task}. Furthermore, an iterative learning strategy is proposed, which iteratively updates the forecasting model parameters by gradient descent with the loss function and updates the loss function based on the newly updated forecasting model. Ultimately, we can achieve value-oriented RES forecasts that are considered ``good'' in terms of the expected overall costs for sequential operations. Our main contributions are,

\textcolor{black}{1) From a practitioner perspective, we develop a computationally efficient framework for training advanced forecasting models (such as Neural Networks) that align with the decision value of a VPP operator. An iterative learning approach is proposed for training, which iteratively updates the forecasting model parameter via gradient descent and the loss function coefficients based on the newly updated forecasts.}

\textcolor{black}{2) From a VPP operator perspective, we explain why value-oriented forecasts are more desirable than quality-oriented forecasts that minimize the mean squared error of renewable forecast and realization. The derived value-oriented loss function indicates that the key reason is the {asymmetric} real-time marginal cost/utility in cases of energy shortage and surplus.}


The remainder is organized as follows: Section \uppercase\expandafter{\romannumeral2} presents the preliminaries regarding the day-ahead and real-time energy dispatch problems with RES. Section \uppercase\expandafter{\romannumeral3} presents the training phase and formulates the parameter estimation method for the value-oriented forecasting model. The iterative solution strategy is proposed in Section \uppercase\expandafter{\romannumeral4}. Sections \uppercase\expandafter{\romannumeral5} and \uppercase\expandafter{\romannumeral6} present case studies with results and discussion, followed by conclusions.

\section{Preliminaries: Day-ahead and Real-time Energy Dispatch}

In this work, we consider the operation of a centralized VPP operator that runs a group of assets including RES, slow-start generators, and some flexible resources (e.g., flexible generators and energy storage). The operator solves an energy dispatch problem, i.e., scheduling the assets to balance total demands at the lowest cost \cite{gomez2018electric}.

Usually, VPPs are operated in a sequential framework, which includes the day-ahead and real-time stages. Decisions in the day-ahead stage are solved at time $t$ on day $d-1$, which schedules the generation of slow-start generators for each time slot, i.e., $\tau,\forall \tau=1,...,T$ on the next day $d$. At the real-time stage, the operator dispatches the flexible resources to compensate for the deficit or surplus of the power imbalance, which is calculated throughout the day $d$, over time intervals of 1 h. In this paper, we assume the real-time dispatch decisions are made, after the RES realization is revealed \footnote{In fact, the realization is approximately informed close to the time $\tau$, as available information is adequate for an accurate estimate.}. \textcolor{black}{Also, the day-ahead schedule is assumed to be fixed at the real-time stage. This is reasonable for inflexible devices, such as combined heat and power (CHP) plants. CHP plants produce both electricity and useful heat simultaneously. The heat demand often dictates their operation, making their electricity generation schedule relatively inflexible \cite{fan2023three}.}

We provide a review of three existing two-stage dispatch models under the uncertainties of RES, whose illustration is shown in Fig. \ref{clarification}. 

1) Deterministic dispatch: The operator schedules slow-start generators at the day-ahead stage, considering a single RES forecast. The real-time dispatch addresses the power imbalance caused by forecasting errors, via flexible resources. This is a simplification of how most centralized operators work.

2) Stochastic dispatch: The operator solves a two-stage stochastic program at the day-ahead stage, which schedules slow-start generators and RES, and anticipates its impact on the real-time operation. In this way, the first-stage decisions of stochastic dispatch involve the decisions regarding slow-start generators and RES schedules. The second-stage decisions, also known as recourse decisions, involve the dispatch of flexible resources under RES scenarios (which is a kind of probabilistic forecast to approximate the conditional distribution). Only the first-stage solutions will be kept at the day-ahead stage. Real-time dispatch is still needed to address the power imbalance caused by the difference between the RES realization and its schedule.

3) Improved deterministic dispatch: As conventional deterministic dispatch is widely known to lead to high overall operating costs, \cite{morales2014electricity} proposed to leverage a strategic RES schedule within the deterministic operation framework, to approach the performance of stochastic dispatch. This schedule is determined by a stochastic bilevel program.

\begin{figure}[h]
  \centering
  \includegraphics[scale=0.38]{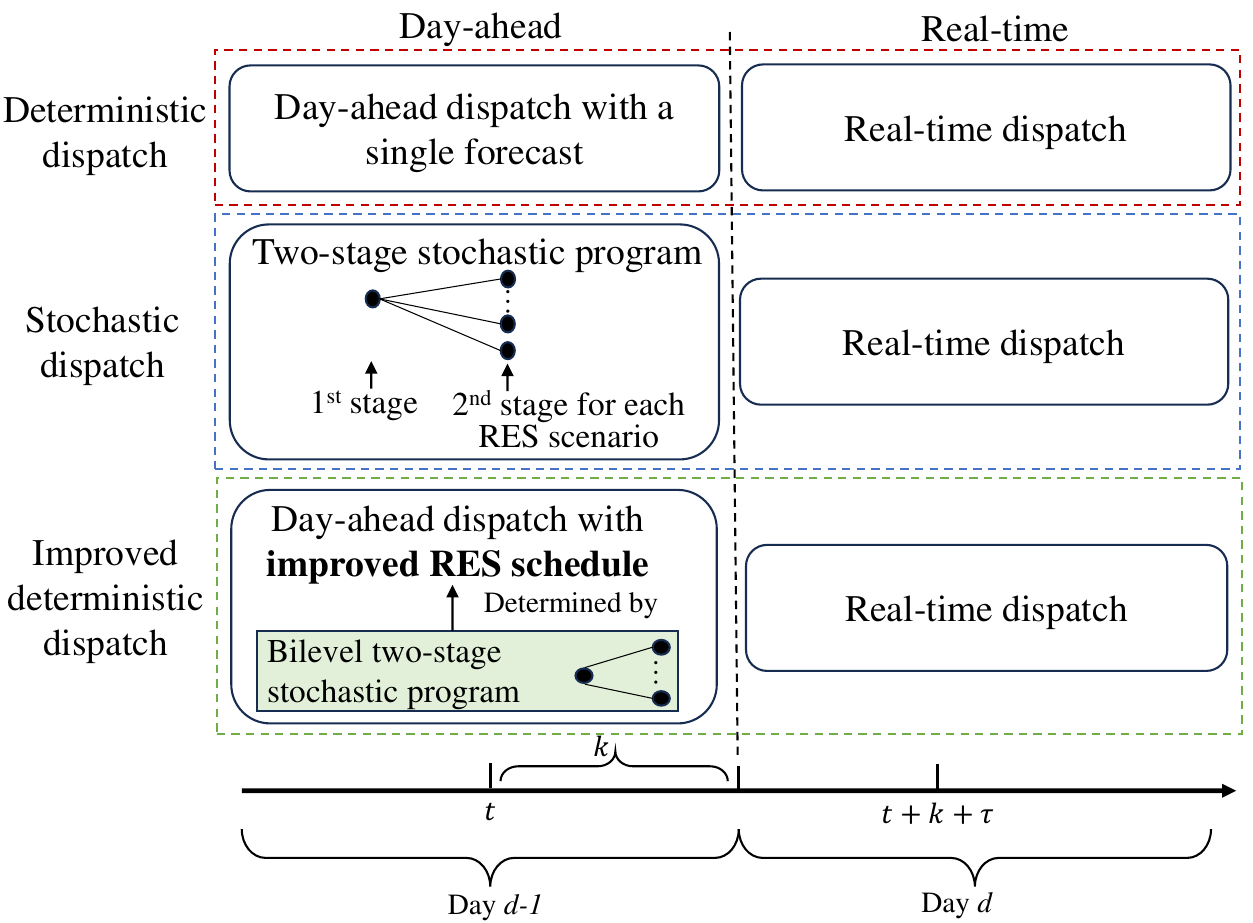}\\
  \caption{The illustration of deterministic, stochastic, and improved deterministic dispatch, where $k$ is the time interval between $t$ and 1 am on day $d$.}
 \label{clarification}
 
\end{figure}

\subsection{Deterministic Dispatch}\label{Conventional Operation} 

\textit{Day-ahead Dispatch with a Point Forecast Predicting Expected Renewable Realization:} Let $y_{d,\tau}$ denote the RES generation at time $\tau$ on day $d$. As it is not known ex-ante, we often use a forecast $\hat{y}_{d,\tau}$ for decision planning. We use $\bm{x}_{d,\tau}$ to denote the decision variables regarding the power generation of slow-start generators at the day-ahead stage. Since the RES has zero marginal cost and we assume that its capacity is less than the minimum load, the operator uses all its forecast $\hat{y}_{d,\tau}$ to balance the load. The problem at the day-ahead stage, minimizing the linear operating cost of slow-start generators in a day, is given by  
\begin{subequations}\label{2}
\begin{alignat}{2} 
&\mathop{\min}_{\{\bm{x}_{d,\tau}\}_{\tau=1}^T}
  &&\sum_{\tau=1}^{T}\bm{\rho}^\top\bm{x}_{d,\tau}\\ 
& \text{s.t.}
&&   f_i^D(\bm{x}_{d,\tau})\leq 0,\forall i \in I^D,\forall \tau=1,...,T \label{2(b)}
    \\ 
&&& 
-\bm{r} \leq \bm{x}_{d,\tau+1}-\bm{x}_{d,\tau} \leq \bm{r}, 
\forall \tau=1,...,T-1\label{2(c)}\\
    &&&  \bm{1}^\top\bm{x}_{d,\tau}+\hat{y}_{d,\tau}=l_{d,\tau},\forall \tau=1,...,T \label{2(d)}.
\end{alignat}
\end{subequations}
\textcolor{black}{where $\bm{\rho}$ is the marginal generation cost.}
We express the optimal solution of the decision $\bm{x}_{d,\tau}$ as $\bm{x}^*_{d,\tau}(\hat{y}_{d,\tau})$, since it depends on the forecast $\hat{y}_{d,\tau}$. \eqref{2(b)} is the constraint regarding the operational limits, such as limiting the output power of the generators, and $I^D$ gives the set of those constraints.
\eqref{2(c)} limits the up- and down-ramping of the generators between adjacent time steps within the limit $\bm{r}$.
The power balance under the load $l_{d,\tau}$ is ensured by \eqref{2(d)}. We assume the load forecast is rather accurate, and the uncertainty is only from RES. 
Additionally, we consider a single-node system.
The total power generation equals the total demand.

\textit{Real-time Dispatch:} As the errors between $\hat{y}_{d,\tau}$ and $y_{d,\tau}$ are unavoidable, it requires the operators to dispatch other assets to compensate the errors as time close to delivery. After the RES realization  $y_{d,\tau}$ is revealed at time $\tau$ on day $d$, the energy imbalance caused by RES is addressed independently at each time $\tau$ in real-time. The energy imbalance is then given by $y_{d,\tau}-\hat{y}_{d,\tau}$,  which represents a surplus of generation, if positive, or a shortage, if negative. We group the decisions taken to cope with the energy shortage into the vector $\bm{z}_{d,\tau}^+,\forall \tau=1,...,T$ and the decisions taken to cope with the energy surplus into the vector $\bm{z}_{d,\tau}^-,\forall \tau=1,...,T$. 
Minimizing the cost of the flexible resources outputs, the problem solved at time $\tau$ for compensating the deviation $y_{d,\tau}-\hat{y}_{d,\tau}$ is,
\begin{subequations}\label{3}
\begin{alignat}{2}  
&\mathop{\min}_{\bm{z}_{d,\tau}^+,\bm{z}_{d,\tau}^-}
&&  \quad \bm{\rho}^{+\top}\bm{z}_{d,\tau}^+-\bm{\rho}^{-\top}\bm{z}_{d,\tau}^-\label{3(a)}\\
    &\text{s.t.} 
    && \quad  f_i^{R+}(\bm{z}_{d,\tau}^+)\leq 0,\forall i \in I^{R+} \label{3(b)}
    \\ 
    &&&  \quad f_i^{R-}(\bm{z}_{d,\tau}^-)\leq 0,\forall i \in I^{R-} \label{3(c)}
    \\
    &&&  \quad \bm{1}^\top(\bm{z}_{d,\tau}^+-\bm{z}_{d,\tau}^-)+y_{d,\tau}-\hat{y}_{d,\tau}=0 \label{3(d)}.
\end{alignat}
\end{subequations}
\textcolor{black}{where $\bm{\rho}^+$ and $\bm{\rho}^-$ are marginal cost and utility.}
We express the optimal solution of the decision $\bm{z}_{d,\tau}^+,\bm{z}_{d,\tau}^-$ as $\bm{z}^{+*}_{d,\tau}(\hat{y}_{d,\tau},y_{d,\tau}),\bm{z}^{-*}_{d,\tau}(\hat{y}_{d,\tau},y_{d,\tau})$, since it depends on the forecast $\hat{y}_{d,\tau}$ and the realization $y_{d,\tau}$. \eqref{3(b)} and \eqref{3(c)} comprise upper and lower bounds
on the flexible resources output, and $I^{R+},I^{R-}$ give the set of those constraints. \eqref{3(d)} ensures that the deviation between RES forecast and realization is addressed, and the system remains in balance. 

\textit{Forecasting Products Requirement:} At the day-ahead stage, operators are required to predict a single value for the RES generation the next day, and then schedule the assets. The forecasts are issued at time $t$ in day-ahead, by a model $g$ with the estimated parameters $\hat{\Theta}_Q$ (the subscript $Q$ stands for quality-oriented forecasts), based on the available information $\bm{s}_d$. Here, $\bm{s}_d$  is a set formed by the realization $\boldsymbol{s}_{d,\tau}$ of the random variable $S_{d,\tau}$ (denoting the contextual information at time slot $\tau$), i.e., $\bm{s}_d = \{\boldsymbol{s}_{d,1},...,\boldsymbol{s}_{d,T}\}$. Take day-ahead wind power forecasting for instance, such information mainly contains numerical weather predictions (NWPs). Denote the forecast for RES at time slot $\tau$ as $\hat{y}_{d,\tau}$, which is obtained by
\begin{equation} \label{1}
    \hat{y}_{d,\tau}=g(\boldsymbol{s}_{d,\tau};\hat{\Theta}_Q).
\end{equation}
One of the most widely accepted practices is issuing $\hat{y}_{d,\tau}$ to predict the ``most likely realization'' of RES random variable $Y_{d,\tau}$. The expected value of $Y_{d,\tau}$ is a natural candidate to take this role. In this way, the parameter $\Theta_Q$ is usually estimated via data-driven methods and particularly by minimizing the mean squared error (MSE) at the training phase.

\subsection{Two-stage Stochastic Program based Dispatch}\label{Two-stage Stochastic Program based Dispatch}

\textit{Day-ahead Dispatch with Probabilistic Forecasts and Stochastic Optimization:} As an input parameter, the forecast $\hat{y}_{d,\tau}$ influences the day-ahead and real-time problems. Concretely, $\hat{y}_{d,\tau}$ affects the generation schedule of slow-start generators, and in turn, affects the day-ahead operating cost. Its over- and down-prediction results in different operating decisions in real-time. In this way, the forecast $\hat{y}_{d,\tau}$ couples the day-ahead and real-time problems. However, the separate operation in \eqref{2} and \eqref{3} fails to take such coupling into account, thereby leading to higher overall operating costs. Therefore, a stochastic program for co-optimizing the two stages is advocated, which considers the temporal dependency of the day-ahead and real-time operation and treats the RES schedule as a decision variable. Such a two-stage stochastic program is solved in day-ahead for all time slots $\tau,\forall \tau=1,...,T$ on the next day $d$. It determines the RES schedule $\Tilde{y}_{d,\tau}$ and generation schedule $\bm{x}_{d,\tau}$, anticipating their influence on the expected real-time cost given any potential realization of the RES. For that, $\Tilde{y}_{d,\tau}$ and $\bm{x}_{d,\tau}$ form the first-stage decisions, and the schedule of $\bm{x}_{d,\tau}$ incurs the first-stage cost $\bm{\rho}^\top\bm{x}_{d,\tau}$. The flexible resources output $\Tilde{\bm{z}}_{d,\tau}^+,\Tilde{\bm{z}}_{d,\tau}^-$ is the second-stage decision, which depends on each realization of RES forecasting error $Y_{d,\tau}-\Tilde{y}_{d,\tau}$. Its dispatch incurs the second-stage cost $\bm{\rho}^{+\top}\Tilde{\bm{z}}_{d,\tau}^+-\bm{\rho}^{-\top}\Tilde{\bm{z}}_{d,\tau}^-$. The stochastic program for minimizing the summation of the first-stage cost and the expected second-stage cost is formulated as,
\begin{subequations}\label{4}
\begin{alignat}{2} 
&\underset{\Xi}{\mathop{\min}}&& \quad \sum_{\tau=1}^{T}\{\bm{\rho}^\top\bm{x}_{d,\tau}+\mathbb{E}_{ Y_{d,\tau}}\left [\bm{\rho}^{+\top}\Tilde{\bm{z}}_{d,\tau}^+-\bm{\rho}^{-\top}\Tilde{\bm{z}}_{d,\tau}^- \right ]\}\label{4(a)}\\
    & \text{s.t.} 
    &&  \quad f_i^D(\bm{x}_{d,\tau})\leq 0,\forall i \in I^D,\forall \tau=1,...,T \label{4(b)}
    \\ 
    &&& \quad -\bm{r} \leq \bm{x}_{d,\tau+1}-\bm{x}_{d,\tau} \leq \bm{r},\forall \tau=1,...,T-1\label{4(c)}\\
    &&& \quad \bm{1}^\top\bm{x}_{d,\tau}+\Tilde{y}_{d,\tau}=l_{d,\tau},\forall \tau=1,...,T \label{4(d)}\\
    &&& \quad 0 \leq \Tilde{y}_{d,\tau} \leq \bar{y},\forall \tau=1,...,T\\
    &&& \quad f_i^{R+}(\Tilde{\bm{z}}_{d,\tau}^+)\leq 0,\forall i \in I^{R+},\forall \tau=1,...,T \label{4(f)}
    \\ 
    &&& \quad f_i^{R-}(\Tilde{\bm{z}}_{d,\tau}^-)\leq 0,\forall i \in I^{R-},\forall \tau=1,...,T \label{4(g)}
    \\
    &&& \quad \bm{1}^\top(\Tilde{\bm{z}}_{d,\tau}^+-\Tilde{\bm{z}}_{d,\tau}^-)+Y_{d,\tau}-\Tilde{y}_{d,\tau}=0,\forall \tau=1,...,T \label{4(h)},
\end{alignat}
\end{subequations}
where $\Xi=\{\bm{x}_{d,\tau},\Tilde{\bm{z}}_{d,\tau}^+,\Tilde{\bm{z}}_{d,\tau}^-,\Tilde{y}_{d,\tau}\}_{\tau=1}^T$. $\bar{y}$ is the capacity of RES and serves as the upper bound for the RES schedule. $Y_{d,\tau}$ follows the conditional distribution $F_{Y_{d,\tau}|S_{d,\tau}=\bm{s}_{d,\tau}}$. One way to estimate the distribution is to predict the independent and identically distributed (i.i.d.) scenarios, with the same probability of occurrence. However, it usually requires a large number of i.i.d. scenarios to accurately model the distribution, which poses a huge computational burden to solving \eqref{4}.

Notice that by anticipating the balancing operation with the constraints \eqref{4(f)},\eqref{4(g)} and \eqref{4(h)} and considering the expected balance cost in the objective \eqref{4(a)}, the stochastic program \eqref{4} empowers the first-stage decisions to account for the impact of the energy imbalance on the second-stage. After solving \eqref{4}, the solution $\boldsymbol{x}_{d,\tau}^*$ will be used for scheduling slow-start generators in day-ahead. 

\textit{Real-time Dispatch:} Similar to the deterministic approach, after the RES realization $y_{d,\tau}$ is revealed, the real-time problem in \eqref{3} is solved for settling the deviation between the realization and the optimal RES schedule, i.e., $y_{d,\tau}-\Tilde{y}_{d,\tau}^*$.

\textit{Forecasting Products Requirement:} Operators are required to predict the conditional distribution $F_{Y_{d,\tau}|S_{d,\tau}=\bm{s}_{d,\tau}}$.

\subsection{Improved Deterministic Dispatch}\label{Improved Schedule of RES}

\textit{Day-ahead Dispatch with Improved RES Schedule:} The two-stage stochastic program \eqref{4} has the objective of minimizing the overall first-stage and second-stage costs. 
In an effort to enhance the performance of separate operation models \eqref{2} and \eqref{3}, we optimize the RES schedule with the same objective as in \eqref{4}, thereby ensuring that this value accounts for the
the second-stage cost resulting from RES uncertainty. Meanwhile, we need to ensure that such a schedule aligns with the deterministic dispatch framework. For that, a bilevel program is formulated \cite{morales2014electricity}, where the upper-level determines the schedule $\Tilde{y}_{d,\tau}$, and the lower-level solves the first-stage and the second-stage problems separately which are equivalent to the conventional dispatch in \eqref{2} and \eqref{3},
\begin{subequations}\label{5}
\begin{alignat}{2}
& \underset{\{\Tilde{y}_{d,\tau}\}_{\tau=1}^T}{\mathop{\min}}&& \sum_{\tau=1}^{T}\{\bm{\rho}^\top\bm{x}_{d,\tau}^*(\Tilde{y}_{d,\tau})+\nonumber\\
&&&\mathbb{E}_{ Y_{d,\tau}}\left [\bm{\rho}^{+\top}\Tilde{\bm{z}}_{d,\tau}^{+*}(\Tilde{y}_{d,\tau},Y_{d,\tau})-\bm{\rho}^{-\top}\Tilde{\bm{z}}_{d,\tau}^{-*}(\Tilde{y}_{d,\tau},Y_{d,\tau}) \right ]\}\label{5a}\\
    & \text{s.t.} 
    &&   0 \leq \Tilde{y}_{d,\tau} \leq \bar{y},\forall \tau=1,...,T\label{5b}\\
    &&&  \{\bm{x}_{d,\tau}^*(\Tilde{y}_{d,\tau})\}_{\tau=1}^T=\mathop{\arg\min}_{\{\bm{x}_{d,\tau}\}_{\tau=1}^T}  \sum_{\tau=1}^{T}\bm{\rho}^\top\bm{x}_{d,\tau}\label{5c}\\
    &&&\text{s.t.} \qquad f_i^D(\bm{x}_{d,\tau})\leq 0,\forall i \in I^D,\forall \tau=1,...,T\label{5d}\\
      &&& \qquad \quad      -\bm{r} \leq \bm{x}_{d,\tau+1}-\bm{x}_{d,\tau} \leq \bm{r},\forall \tau=1,...,T-1\label{5e}\\
      &&&  \qquad \quad \bm{1}^\top\bm{x}_{d,\tau}+\Tilde{y}_{d,\tau}=l_{d,\tau},\forall \tau=1,...,T\label{5f}\\
      &&&  \{\Tilde{\bm{z}}_{d,\tau}^{+*}(\Tilde{y}_{d,\tau},Y_{d,\tau}),\Tilde{\bm{z}}_{d,\tau}^{-*}(\Tilde{y}_{d,\tau},Y_{d,\tau})= \nonumber \\
      &&&\mathop{\arg\min}_{\Tilde{\bm{z}}^+_{d,\tau},\Tilde{\bm{z}}^-_{d,\tau}} \bm{\rho}^{+\top}\Tilde{\bm{z}}_{d,\tau}^+-\bm{\rho}^{-\top}\Tilde{\bm{z}}_{d,\tau}^-\label{5g}\\
    &&& \text{s.t.} \qquad f_i^{R+}(\Tilde{\bm{z}}_{d,\tau}^+)\leq 0,\forall i \in I^{R+}\label{5h}\\
    &&& \qquad \quad f_i^{R-}(\Tilde{\bm{z}}_{d,\tau}^-)\leq 0,\forall i \in I^{R-}\label{5i}\\
      &&& \qquad \quad     \bm{1}^\top(\Tilde{\bm{z}}_{d,\tau}^+-\Tilde{\bm{z}}_{d,\tau}^-)+Y_{d,\tau}-\Tilde{y}_{d,\tau}=0\}\nonumber,\\
      &&& \qquad \quad\forall \tau=1,...,T,\label{5j}
\end{alignat}
\end{subequations}
where the random variable $Y_{d,\tau}$ follows the conditional distribution $F_{Y_{d,\tau}|S_{d,\tau}=\bm{s}_{d,\tau}}$. The upper-level decision $\Tilde{y}_{d,\tau}$ serves as the parameter in the lower-level problems, and affects the lower-level decisions $\bm{x}_{d,\tau}^*(\Tilde{y}_{d,\tau})$ and $\bm{z}_{d,\tau}^{+*}(\Tilde{y}_{d,\tau},Y_{d,\tau}),\bm{z}_{d,\tau}^{-*}(\Tilde{y}_{d,\tau},Y_{d,\tau})$, which form the upper-level objective in \eqref{5}. After solving \eqref{5}, the improved RES schedule $\Tilde{y}_{d,\tau}^*,\forall \tau=1,...,T$ will enter the day-ahead problem in \eqref{2} to take the role of RES forecast $\hat{y}_{d,\tau}$.

\textit{Real-time Dispatch:} Real-time dispatch is the same as that in Section \ref{Two-stage Stochastic Program based Dispatch}.

\textit{Forecasting Products Requirement:} To solve the stochastic program in \eqref{5}, scenarios for approximating the conditional distribution is needed. Therefore, forecasting products is the same as that in Section \ref{Two-stage Stochastic Program based Dispatch}.

Although the program \eqref{5} can determine the RES schedule by accounting for the future balancing cost caused by RES uncertainty, it requires resolving the stochastic bilevel program every time for determining the RES schedule on a new day, which is not computationally efficient. To save the computational burden, we will learn a forecasting model in the next section, which maps the realization of contextual information to the improved RES schedule.

\section{Methodology}

In this section, we propose to train a forecasting model $g$ to issue the strategic RES schedule, given the contextual information $\bm{s}_{d,\tau}$.
To this end, we align the objective of the parameter estimation with the same objective of \eqref{5}, i.e., minimizing the expected operating costs of the two stages, and term such an approach as value-oriented forecasting. Similar to the commonly used forecasts defined in \eqref{1}, we still model forecasts with a function $g$ with parameters $\Theta_V$ (to make a difference with the model \eqref{1}, the subscript $V$ stands for the proposed value-oriented forecasts). However, the goal here is not to forecast the expectation of the random variable $Y_{d,\tau}$, but a \emph{strategic} quantity to be used at the day-ahead stage. Given the contextual information $\boldsymbol{s}_{d,\tau}$, the value-oriented RES forecast is derived as
\begin{equation}\label{6}
   \Tilde{y}_{d,\tau}=g(\bm{s}_{d,\tau};
   \Theta_V). 
\end{equation}

Usually, the parameters $\Theta_V$ are estimated via a training set formed by historical contextual information and RES realization in $D$ days, i.e., $\{\{\bm{s}_{d,\tau},y_{d,\tau}\}_{\tau=1}^T\}_{d=1}^D$. The data in the training set are assumed to be independently and identically drawn from a joint distribution $F_{Y_{d,\tau},S_{d,\tau}}$ over $Y_{d,\tau}$ and $S_{d,\tau}$. In this way, for any fixed contextual information $S_{d,\tau}=\bm{s}_{d,\tau}$, the conditional distribution $F_{Y_{d,\tau}|S_{d,\tau}=\bm{s}_{d,\tau}}$ is approximated by a single support $y_{d,\tau}$ \cite{bishop2006pattern}. The ultimate estimate for the parameters $\Theta_V$ is denoted as $\hat{\Theta}_V$.
Particularly, at the training phase and given the context $\bm{s}_{d,\tau}$, we denote the forecasts given by $g$ as $\Tilde{y}_{d,\tau}$, with $\Theta_V$ given by any value.
We have the following assumption:

\begin{assumption}
    The day-ahead operation in \eqref{2} and the real-time operation in \eqref{3} have unique primal and dual solutions.
\end{assumption}

The above assumption is reasonable. The primal solutions are interpreted as the scheduled energy, while the dual solutions represent the corresponding prices. Typically, these solutions are unique in the operation. In what follows, we formulate the parameter estimation problem at the training phase in Section \ref{Training Phase}, followed by the operation phase in Section \ref{Operational Forecasting Phase}.

\begin{figure}[th]
  \centering
  \includegraphics[scale=0.55]{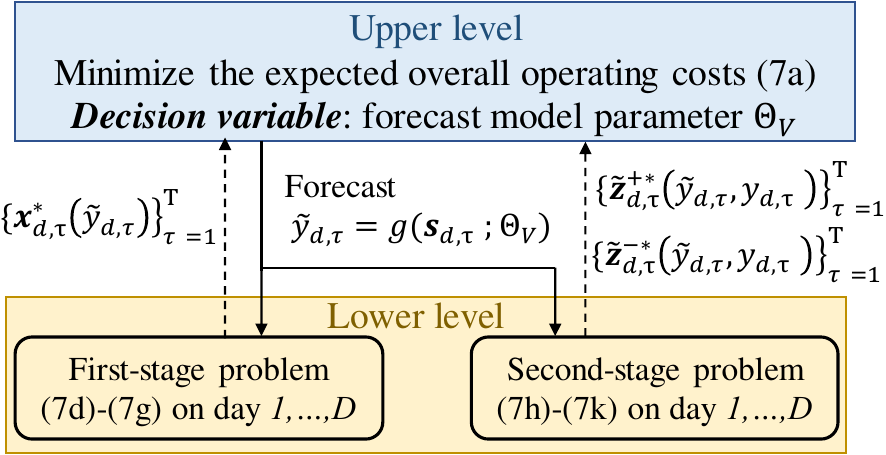}\\
  \caption{Illustration of the proposed bilevel program \eqref{7} for model training of value-oriented forecasting. \vspace{-0.5cm}}
  \label{bilevel}
\end{figure}
\subsection{Training Phase} \label{Training Phase}
For parameter estimation at the training phase, we propose formulating a bilevel program, adopting a similar bilevel form as \eqref{5}. The upper-level problem determines the parameter $\Theta_V$, for minimizing the average of the overall operating cost over $D$ days. Specifically, since the conditional distribution  $F_{Y_{d,\tau}|S_{d,\tau}=\bm{s}_{d,\tau}}$ is supported only by $y_{d,\tau}$, the expected second-stage cost $\mathbb{E}_{Y_{d,\tau}}[\bm{\rho}^{+\top}\Tilde{\bm{z}}_{d,\tau}^{+*}(\Tilde{y}_{d,\tau},Y_{d,\tau})-\bm{\rho}^{-\top}\Tilde{\bm{z}}_{d,\tau}^{-*}(\Tilde{y}_{d,\tau},Y_{d,\tau})]$ is approximated by the cost under this support, i.e., $\bm{\rho}^{+\top}\Tilde{\bm{z}}_{d,\tau}^{+*}(\Tilde{y}_{d,\tau},y_{d,\tau})-\bm{\rho}^{-\top}\Tilde{\bm{z}}_{d,\tau}^{-*}(\Tilde{y}_{d,\tau},y_{d,\tau})$. The lower-level repeats first- and second-stage problems for each day, which amounts to the day-ahead and real-time problems in \eqref{2} and \eqref{3}. The optimal solutions parameterized by the forecast $\Tilde{y}_{d,\tau}$ are obtained from the lower level. The bilevel program is,
\begin{subequations}\label{7}
\begin{alignat}{2} 
&\underset{{\Theta_V}}{\mathop{\min}}&& \quad \frac{1}{D\cdot T}\sum_{d=1}^{D}\sum_{\tau=1}^{T}\{\bm{\rho}^\top\bm{x}_{d,\tau}^*(\Tilde{y}_{d,\tau})+\nonumber\\
&&& \bm{\rho}^{+\top}\Tilde{\bm{z}}_{d,\tau}^{+*}(\Tilde{y}_{d,\tau},y_{d,\tau})-\bm{\rho}^{-\top}\Tilde{\bm{z}}_{d,\tau}^{-*}(\Tilde{y}_{d,\tau},y_{d,\tau})\}\label{7a}\\
    & \text{s.t.} 
    && \Tilde{y}_{d,\tau}=g(\bm{s}_{d,\tau};\Theta_V),\forall \tau=1,...,T,\forall d=1,...,D\\ 
    &&&0 \leq \Tilde{y}_{d,\tau} \leq \bar{y},\forall \tau=1,...,T,\forall d=1,...,D\\
    &&&\{\{\bm{x}_{d,\tau}^*(\Tilde{y}_{d,\tau})\}_{\tau=1}^T=\mathop{\arg\min}_{\{\bm{x}_{d,\tau}\}_{\tau=1}^T}  \sum_{\tau=1}^{T}\bm{\rho}^\top\bm{x}_{d,\tau}\label{7d}\\
    &&&\text{s.t.} \qquad f_i^D(\bm{x}_{d,\tau})\leq 0:\delta_{i,d,\tau},\forall i \in I^D,\forall \tau=1,...,T\label{7e}\\
      &&& \qquad \quad      -\bm{r} \leq \bm{x}_{d,\tau+1}-\bm{x}_{d,\tau} \leq \bm{r}:\underline{\bm{\eta}_{d,\tau}},\overline{\bm{\eta}_{d,\tau}},\nonumber\\
      &&& \qquad \qquad \qquad \qquad \qquad \qquad   \forall \tau=1,...,T-1\label{7f}\\
      &&&  \qquad \quad \bm{1}^\top\bm{x}_{d,\tau}+\Tilde{y}_{d,\tau}=l_{d,\tau}:\lambda_{d,\tau},\forall \tau=1,...,T\},\nonumber\\
      &&& \qquad \qquad \qquad \qquad \qquad \qquad \qquad \quad  \forall d=1,...,D\label{7g}\\
    &&&  \{\Tilde{\bm{z}}_{d,\tau}^{+*}(\Tilde{y}_{d,\tau},y_{d,\tau}),\Tilde{\bm{z}}_{d,\tau}^{-*}(\Tilde{y}_{d,\tau},y_{d,\tau})= \nonumber\\
    &&&\mathop{\arg\min}_{\Tilde{\bm{z}}_{d,\tau}^+,\Tilde{\bm{z}}_{d,\tau}^-} \bm{\rho}^{+\top}\Tilde{\bm{z}}_{d,\tau}^+-\bm{\rho}^{-\top}\Tilde{\bm{z}}_{d,\tau}^-\label{7h}\\
    &&& \text{s.t.} \qquad f_i^{R+}(\Tilde{\bm{z}}_{d,\tau})\leq 0:\mu_{i,d,\tau},\forall i \in I^{R+}\label{7i}\\
    &&& \qquad \quad f_i^{R-}(\Tilde{\bm{z}}_{d,\tau})\leq 0:\zeta_{i,d,\tau},\forall i \in I^{R-}\label{7j}\\
      &&& \qquad \quad     \bm{1}^\top(\Tilde{\bm{z}}_{d,\tau}^+-\Tilde{\bm{z}}_{d,\tau}^-)+y_{d,\tau}-\Tilde{y}_{d,\tau}=0:\nu_{d,\tau}\},\nonumber\\
      &&& \qquad \qquad \qquad \qquad \quad \forall \tau=1,...,T,\forall d=1,...,D,\label{7k}
\end{alignat}
\end{subequations}
where the term inside the curly bracket in \eqref{7a} is the expected day-ahead and real-time overall operating cost at the time slot $\tau$ on day $d$. The model parameter $\Theta_V$, which is the upper-level variable, has an influence on the forecast $\Tilde{y}_{d,\tau}$, and also the parameterized solutions $\bm{x}_{d,\tau}^*(\Tilde{y}_{d,\tau}),\Tilde{\bm{z}}_{d,\tau}^{+*}(\Tilde{y}_{d,\tau},y_{d,\tau}),\Tilde{\bm{z}}_{d,\tau}^{-*}(\Tilde{y}_{d,\tau},y_{d,\tau})$ as well as the average operating cost over $D$ days through \eqref{7a}.  
An illustration of such a bilevel program is shown in Fig. \ref{bilevel}.

\begin{remark} 
 As one of the inspiring works in the literature regarding the application of bilevel program for parameter estimation, \cite{morales2023prescribing} used a simplified first-stage dispatch model. We go beyond \cite{morales2023prescribing}. Our model \eqref{7} can address the temporal correlated constraints, such as ramping constraints in \eqref{7f},  which aligns with the current operation practice.
\end{remark}


The optimal dual solutions of \eqref{7d}-\eqref{7g} parameterized by $\Tilde{y}_{d,\tau}$ are denoted as $\delta_{i,d,\tau}^*(\Tilde{y}_{d,\tau})$, $\underline{\bm{\eta}_{d,\tau}}^*(\Tilde{y}_{d,\tau})$, $\overline{\bm{\eta}_{d,\tau}}^*(\Tilde{y}_{d,\tau})$, and $\lambda_{d,\tau}^*(\Tilde{y}_{d,\tau})$. The optimal dual solutions of \eqref{7h}-\eqref{7k} parameterized by $\Tilde{y}_{d,\tau},y_{d,\tau}$ are denoted as $\mu_{i,d,\tau}^*(\Tilde{y}_{d,\tau},y_{d,\tau})$, $\zeta_{i,d,\tau}^*(\Tilde{y}_{d,\tau},y_{d,\tau})$ and $\nu_{d,\tau}^*(\Tilde{y}_{d,\tau},y_{d,\tau})$.
Usually, the operational constraints are linear, and the problems with linear operating cost in \eqref{7d}-\eqref{7g} and \eqref{7h}-\eqref{7k} are linear programs. In this way, we resort to the dual problems for explicitly rewriting  \eqref{7a} as a function of the forecasting model output $\Tilde{y}_{d,\tau}$. We have the following proposition for it.

\begin{prop}
 Consider the convex optimization problems \eqref{7d}-\eqref{7g} and \eqref{7h}-\eqref{7k} in the form of linear programs (LPs), with the parameterized dual solutions associated with the constraints \eqref{7g} and \eqref{7k} are $\lambda_{d,\tau}^*(\Tilde{y}_{d,\tau})$ and $\nu_{d,\tau}^*(\Tilde{y}_{d,\tau},y_{d,\tau})$, respectively. The optimal primal objectives of \eqref{7d} and \eqref{7h} at time $\tau$ on day $d$ are $\bm{\rho}^\top\bm{x}_{d,\tau}^*(\Tilde{y}_{d,\tau})$ and $\bm{\rho}^{+\top}\Tilde{\bm{z}}_{d,\tau}^{+*}(\Tilde{y}_{d,\tau},y_{d,\tau})-\bm{\rho}^{-\top}\Tilde{\bm{z}}_{d,\tau}^{-*}(\Tilde{y}_{d,\tau},y_{d,\tau})$.
As the strong duality holds, the optimal primal objectives equal the respective optimal dual objectives, expressed by the optimal dual solutions. The term in the curly bracket of \eqref{7a} is equivalent as,
\begin{equation}\label{12}
-\lambda_{d,\tau}^*(\Tilde{y}_{d,\tau}) \Tilde{y}_{d,\tau}-\nu_{d,\tau}^* (\Tilde{y}_{d,\tau},y_{d,\tau})(y_{d,\tau}-\Tilde{y}_{d,\tau})+ \psi^{D,*}_{d,\tau}+\psi^{R,*}_{d,\tau},
\end{equation}
where the notations $\psi^{D,*}_{d,\tau}$ and $\psi^{R,*}_{d,\tau}$ denote the remaining terms of the optimal dual objectives, which no longer involve $\Tilde{y}_{d,\tau}$.
\end{prop}
\begin{proof}
    See Appendix \ref{Appendix A}.
\end{proof}


Equation \eqref{12} shows that the model parameter estimation at the upper-level problem is associated with the parameterized dual solutions provided by the lower-level problem. Therefore, the primal problems at the lower level are replaced with their dual problems, denoted as $\text{DP}^D(\Tilde{y}_{d,\tau}),\text{DP}^R(\Tilde{y}_{d,\tau},y_{d,\tau})$, such that the parameterized dual solutions can be obtained. Therefore,   the primal problems in \eqref{7d}-\eqref{7g} and \eqref{7h}-\eqref{7k} are replaced with 
\begin{equation}\label{13}
\begin{aligned}
    &\delta_{i,d,\tau}^*(\Tilde{y}_{d,\tau}),\underline{\bm{\eta}_{d,\tau}}^*(\Tilde{y}_{d,\tau}),\overline{\bm{\eta}_{d,\tau}}^*(\Tilde{y}_{d,\tau}),\lambda_{d,\tau}^*(\Tilde{y}_{d,\tau}) =\\ & \mathop{\arg\max}_{\delta_{i,d,\tau} \geq 0,\underline{\bm{\eta}_{d,\tau}} \geq 0,\overline{\bm{\eta}_{d,\tau}} \geq 0,\lambda_{d,\tau}}\text{DP}^D(\Tilde{y}_{d,\tau}),\\
    & \mu_{i,d,\tau}^*(\Tilde{y}_{d,\tau},y_{d,\tau}),\zeta_{i,d,\tau}^*(\Tilde{y}_{d,\tau},y_{d,\tau}),\nu_{d,\tau}^*(\Tilde{y}_{d,\tau},y_{d,\tau}) =\\ & \mathop{\arg\max}_{\mu_{i,d,\tau} \geq 0,\zeta_{i,d,\tau} \geq 0,\nu_{d,\tau}}\text{DP}^R(\Tilde{y}_{d,\tau},y_{d,\tau}).
\end{aligned}
\end{equation}

With the upper-level problem defined in \eqref{12} and the lower-level problem defined in \eqref{13}. The bilevel program in \eqref{7} can be equivalently written as,
\begin{subequations}\label{14}
\begin{alignat}{2}
     &\mathop{\min}_{\Theta_V} \quad
     &&  \frac{1}{D\cdot T}\sum_{d=1}^D\sum_{\tau=1}^T \eqref{12}\label{14a}\\
     & \text{s.t.} \quad &&  \Tilde{y}_{d,\tau}=g(\bm{s}_{d,\tau};\Theta_V)\\
     &&& 0 \leq \Tilde{y}_{d,\tau} \leq \bar{y}\\
    &&&   \eqref{13}.
\end{alignat}
\end{subequations}

As can be seen from the structure of the problem, where the lower-level problem (the follower) in \eqref{13} optimizes its objective treating the forecast $\Tilde{y}_{d,\tau}$ as the parameter, and the goal of the upper-level problem (the leader) is to determine $\Tilde{y}_{d,\tau}$ that minimizes the objective which involves the parameterized solutions returned by the lower level. 
A natural idea is to resort to an iterative solution for model parameter estimation. The details are given in Section \ref{Iterative Algorithm for Parameter Estimation}.

\subsection{Operational Forecasting Phase}\label{Operational Forecasting Phase}
With the estimated parameter $\hat{\Theta}_V$, the forecasting model is ready for use during operation. Usually at 12 pm on day $d-1$ when the decisions are made at the day-ahead stage, with the available information $\boldsymbol{s}_{d,\tau},\forall \tau=1,...,T$, we forecast the RES generation for each time slot $\tau$, i.e., 
\begin{equation*} \Tilde{y}_{d,\tau}=g(\boldsymbol{s}_{d,\tau};\hat{\Theta}_V), \ \tau = 1,\cdots, T.
\end{equation*}
Then, $\forall \tau=1,...,T$, $\Tilde{y}_{d,\tau}$ is fed into the day-ahead problem \eqref{2} as a parameter. The real-time problem in \eqref{3} is solved when the realization $y_{d,\tau}$ is available at time slot $\tau$ on day $d$. 

\section{Iterative Algorithm for Parameter Estimation}\label{Iterative Algorithm for Parameter Estimation}


We present an iterative learning strategy for the parameter estimation problem \eqref{14} in Section \ref{Solution Strategy}. In Section \ref{interpretate}, we explain why value-oriented forecasts are more preferable to VPP operators via the derived loss function.

\subsection{Iterative Learning Algorithm}\label{Solution Strategy}
In the iterative learning strategy, the upper level issues the forecast $\Tilde{y}_{d,\tau}^{e-1}=g(\bm{s}_{d,\tau};\Theta_V^{e-1})$ at the epoch $e$, where the model parameter is denoted as $\Theta_V^{e-1}$. Then, the forecast $\Tilde{y}_{d,\tau}^{e-1}$ is passed to the lower level. The lower level solves the dual problems \eqref{13} and passes the optimal dual solutions parameterized by the forecast $\Tilde{y}_{d,\tau}^{e-1}$ to the upper level. Such dual solutions update the coefficients in upper-level objectives. The upper level updates the model parameter with the updated objective. We denote the updated parameter at the epoch $e$ as $\Theta_V^{e}$.  The forecast $\Tilde{y}_{d,\tau}^e$ is then redetermined with the updated parameter $\Theta_V^{e}$, and passed to the lower level. Such a process repeats for the epochs $e=1,2,...,$. The parameter $\Theta^0_V$ in the first epoch is randomly initialied. We denote the updated parameter at the last epoch as $\hat{\Theta}_V$, which will be used in the operational forecasting phase. Next, we detail the upper-level problem in the iterative scheme.


\begin{figure*}[h]
  \centering
  \includegraphics[height=4cm]{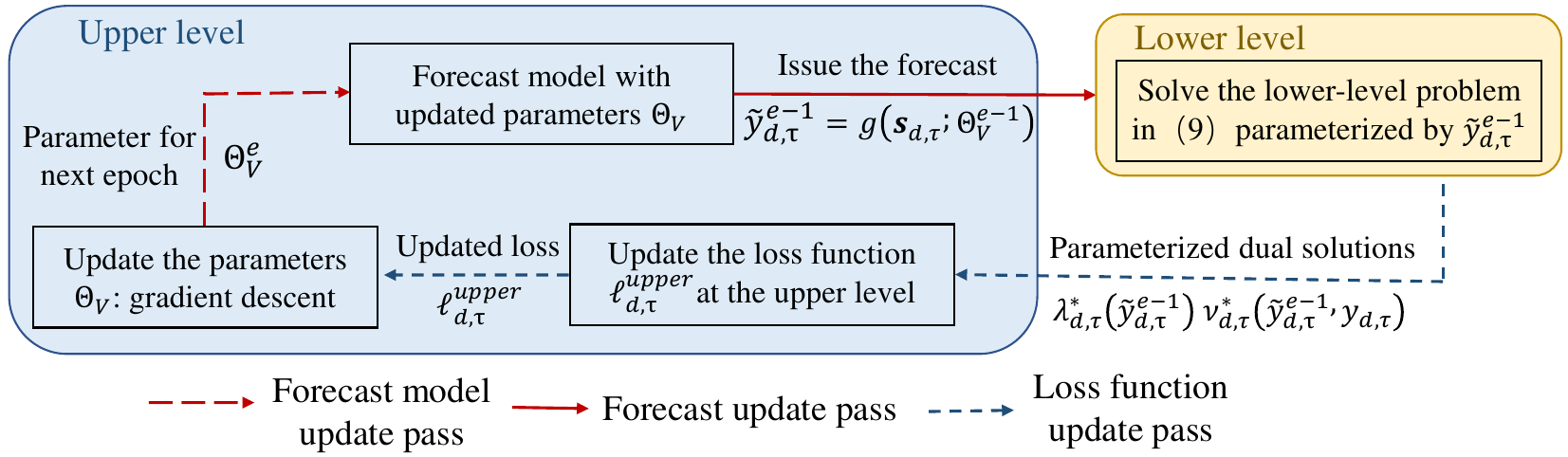}\\
  \caption{The illustration of iterative scheme at the epoch $e$.}\label{Fig 3}
\end{figure*}


The parameterized dual solutions provided by the lower-level problem at the epoch $e$ are denoted as  $\delta_{i,d,\tau}^*(\Tilde{y}_{d,\tau}^{e-1})$, $\underline{\bm{\eta}_{d,\tau}}^*(\Tilde{y}_{d,\tau}^{e-1})$, $\overline{\bm{\eta}_{d,\tau}}^*(\Tilde{y}_{d,\tau}^{e-1})$, $\lambda^*_{d,\tau}(\Tilde{y}_{d,\tau}^{e-1})$, $\mu^*_{i,d,\tau}(\Tilde{y}_{d,\tau}^{e-1},y_{d,\tau})$, $\zeta_{i,d,\tau}^*(\Tilde{y}_{d,\tau}^{e-1},y_{d,\tau})$, and $\nu^*_{d,\tau}(\Tilde{y}_{d,\tau}^{e-1},y_{d,\tau})$. These dual solutions are treated as the constants at the upper-level problem under the iterative scheme, since the relationship between them and the forecasts is a constant function in the neighborhood of the forecast $\Tilde{y}_{d,\tau}^{e-1}$ \cite{zheng2024distributionally,borrelli2003geometric}. Therefore, the last two terms $\psi^{D,*}_{d,\tau}$ and $\psi^{R,*}_{d,\tau}$ in \eqref{12} are constants as they only include the parameterized dual solutions. Thus, the upper-level objective in \eqref{12} can be further simplified by removing those terms, and the upper-level problem is reformulated as,
\begin{equation}\label{16}
\begin{aligned}
     &\mathop{\min}_{\Theta_V}
     &&
     \frac{1}{D\cdot T}\sum_{d=1}^D\sum_{\tau=1}^T-\lambda_{d,\tau}^*(\Tilde{y}_{d,\tau}^{e-1}) \Tilde{y}_{d,\tau}-\\
     &&&\nu_{d,\tau}^* (\Tilde{y}_{d,\tau}^{e-1},y_{d,\tau})(y_{d,\tau}-\Tilde{y}_{d,\tau})\\
     & \text{s.t.} &&  
 \Tilde{y}_{d,\tau}=g(\bm{s}_{d,\tau};\Theta_V)\\
 &&& 0 \leq  \Tilde{y}_{d,\tau} \leq \bar{y},
\end{aligned}
\end{equation}
whose objective is a linear function of $\Tilde{y}_{d,\tau}$ in the neighborhood where the coefficients $\lambda_{d,\tau}^*(\Tilde{y}_{d,\tau}^{e-1})$, $\nu_{d,\tau}^* (\Tilde{y}_{d,\tau}^{e-1},y_{d,\tau})$ remain unchanged. Such a locally linear objective acts as the value-oriented loss function for guiding the parameter estimation.
For a sample in the training set, we denote the value-oriented loss function as $\ell_{d,\tau}^{upper}$, i.e.,
\begin{equation}\label{loss_up}
\ell_{d,\tau}^{upper}:=-\lambda_{d,\tau}^*(\Tilde{y}_{d,\tau}^{e-1}) \Tilde{y}_{d,\tau}-\nu_{d,\tau}^* (\Tilde{y}_{d,\tau}^{e-1},y_{d,\tau})(y_{d,\tau}-\Tilde{y}_{d,\tau})  
\end{equation}


Therefore, the upper-level objective can be rewritten as $\frac{1}{D\cdot T}\sum_{d=1}^D\sum_{\tau=1}^T\ell_{d,\tau}^{upper}$.

With the defined upper- and lower-level problems, the training under the iterative scheme can be performed in a fixed number of epochs. 
Although various types of regression models can be used here as a forecasting model, we use neural network (NN) with batch optimization as an example, and its training paradigm is summarized in Algorithm \ref{alg1}. During training, with the defined loss function $\ell_{d,\tau}^{upper}$, the forecasting model parameter $\Theta_V$ is updated via gradient descent at the upper level in epoch $e$, i.e.,
\begin{equation}\label{gd}
    \Theta^e_V \leftarrow \Theta^{e-1}_V-\alpha \bigtriangledown_{\Theta_V}\frac{1}{B\cdot T}\sum_{d=1}^B\sum_{\tau=1}^T\ell^{upper}_{d,\tau}
\end{equation}

The forecasts are issued by the updated forecasting model and passed to the lower level. Then, the dual solutions are obtained from the lower level and update the loss function at the upper level. Such a process repeats for several epochs. An illustration is shown in Fig. \ref{Fig 3}.

\subsection{Interpretating Value-oriented Forecasts}\label{interpretate}

\textcolor{black}{We would like to interpret the forecasts with the designed loss function \eqref{loss_up}.}
\textcolor{black}{Its coefficients $\lambda_{d,\tau}^*(\Tilde{y}_{d,\tau}^{e-1})$, $\nu_{d,\tau}^* (\Tilde{y}_{d,\tau}^{e-1},y_{d,\tau})$ are the respective optimal dual solutions associated with the equality constraint in day-ahead and real-time problems. They can be interpreted as electricity prices and are associated with the day-ahead marginal cost and real-time marginal cost/utility. In the following, we omit the parentheses and denote them as $\lambda_{d,\tau}^*,\nu_{d,\tau}^*$ for simplicity.}

\textcolor{black}{The loss function in \eqref{loss_up} can be equivalently written as,}
\begin{equation}\label{reloss}
    -\lambda_{d,\tau}^* y_{d,\tau}+(\lambda_{d,\tau}^*-\nu_{d,\tau}^*)(y_{d,\tau}-\Tilde{y}_{d,\tau})
\end{equation}

\textcolor{black}{Since the first term $-\lambda_{d,\tau}^*y_{d,\tau}$ is a constant, minimizing \eqref{reloss} is equivalent with minimizing the second term,}
\begin{equation}\label{reloss1}
    (\lambda_{d,\tau}^*-\nu_{d,\tau}^*)(y_{d,\tau}-\Tilde{y}_{d,\tau})
\end{equation}
\textcolor{black}{where $\lambda_{d,\tau}^*-\nu_{d,\tau}^*$ can be interpreted as the difference between day-ahead marginal cost and real-time marginal cost/utility. The asymmetric real-time marginal cost/utility can cause the value $\lambda_{d,\tau}^*-\nu_{d,\tau}^*$ different in the cases of energy deficit and energy surplus, even if the day-ahead marginal cost takes the same value. For instance, when the real-time marginal cost of flexible resources addressing energy shortage (whose resulted real-time dual solution is denoted as $\nu_{d,\tau}^{+*}$) is larger than the real-time marginal utility of flexible resources addressing energy surplus (whose resulted real-time dual solution is denoted as $\nu_{d,\tau}^{-*}$) and satisfying $\nu_{d,\tau}^{+*} - \lambda_{d,\tau}^{*} \geq  \lambda_{d,\tau}^{*}-\nu_{d,\tau}^{-*} \geq 0$, the proposed approach tends to forecast less renewable energy. The reason is that per forecast deviation in the case of energy shortage ($y_{d,\tau} < \Tilde{y}_{d,\tau}$) results in more loss than in the case of energy surplus ($y_{d,\tau} > \Tilde{y}_{d,\tau}$). The quality-oriented forecasts predicting the expected realization treat the two cases equally, ignoring the different impacts on the cost. Therefore, the quality-oriented forecasts result in more cost and are less preferable to the operator.} 

\begin{algorithm}
	\caption{Training Phase of a NN-based forecasting model for Value-oriented Forecasting}
	\label{alg1}
	\begin{algorithmic}[1]
	\Require{Learning rate $\alpha$, batch size $B$, and initialized forecasting model parameters $\Theta_V^0$}
	\For{epoch $e=1,2,\dots$}
	
        \State{Sample a batch $\{\{\bm{s}_{d,\tau},y_{d,\tau}\}_{\tau=1}^T\}_{d=1}^B$ from the training set}
	    \State{Issue forecasts $\Tilde{y}_{d,\tau}^{e-1}=g(\bm{s}_{d,\tau};\Theta_V^{e-1})$, $\forall \tau=1,\cdots,T$, $d=1,\cdots,B$ via the forecasting model $g$.}
	  
	    \For{batch $d=1,2,\dots,B$}
	    \State{Obtain dual solutions $\{\lambda_{d,\tau}^*(\Tilde{y}_{d,\tau}^{e-1})\}_{\tau=1}^T$ and $\{\nu_{d,\tau}^*(\Tilde{y}_{d,\tau}^{e-1},y_{d,\tau})\}_{\tau=1}^T$ as the formula (\ref{13})}
	   \EndFor
	   \State{Update the loss function with $\{\{\lambda_{d,\tau}^*(\Tilde{y}_{d,\tau}^{e-1})\}_{\tau=1}^T\}_{d=1}^B$ and $\{\{\nu_{d,\tau}^*(\Tilde{y}_{d,\tau}^{e-1},y_{d,\tau})\}_{\tau=1}^T\}_{d=1}^B$ as the formula \eqref{loss_up}.}
    \State{Update model parameters via gradient descent \eqref{gd}.}
    \EndFor
\end{algorithmic} \end{algorithm}

\section{Experiement Setup}
This section considers the sequential deterministic operational problem faced by a VPP operator managing wind power and analyzes the effectiveness of value-oriented forecasting in this context. A thorough examination of performance in the test set is carried out, focusing on five key aspects: 1) 
We will demonstrate that our proposed approach leads to lower average operating costs compared to the quality-oriented forecasting approach, which employs MSE or pinball loss as the loss functions during training; 2) We will uncover the operational advantages under large wind power penetrations by investigating the performances at different levels of wind power capacities; 3) 
The computational efficiency of the proposed approach will be demonstrated by discussing the computation time in the test set; \textcolor{black}{4) We will discuss the potential of extending the proposed approach to operational problems involving binary variables; and 5) We will discuss the potential for the proposed approach to adapt to operational changes. Concretely, the changes are mainly related to parameter alterations rather than changes in the structure or type of the problem.}

\subsection{Description of Operation Simulations and Data}
In the context of the energy dispatch problem, we consider the operation of a centralized VPP operator with two slow-start distributed generators (DGs) and two flexible resources. DGs and wind power generators are considered at the day-ahead stage. During operation, the operator first solves the day-ahead problem \eqref{2}. After the ground truth of wind power generation is revealed, the real-time operation problem \eqref{3} is solved. The detailed model and parameters are given in \cite{News}. \textcolor{black}{In the real-time operation problem, if an energy surplus occurs (where the realization of wind power is greater than the forecast), the centralized operator will use flexible resources to absorb excessive power and obtain utility from the usage of electricity. Therefore, the real-time operating cost will be negative. If an energy shortage occurs (where wind power production is less than forecast), the centralized operator will leverage flexible resources to produce power that solves the shortage. Therefore, the real-time operating cost will be positive.}

The hourly wind power production in 2012 from GEFCom 2014 is used, along with the aggregated annual load data \textcolor{black}{from the Low Carbon London dataset, which has the valley demand of 50 \unit{kW}, the average demand of 56 \unit{kW}, and the peak demand of 70 \unit{kW}. Wind and load data are provided in \cite{News}}. 80\% data is divided into the training set, while the rest forms the test set. The wind data are scaled by multiplying a constant to fit the parameter setting of the case study. The contextual information consists of the predicted wind speed and direction at altitudes of 10 meters and 100 meters.

\subsection{Benchmark Models}
\textcolor{black}{In general, we use six benchmark models for comparison: two quality-oriented (\textbf{Qua-E} and \textbf{Qua-Q}), two value-oriented (\textbf{Val-L} and \textbf{Val-O}) forecasting models, perfect forecasts (\textbf{Per-F}), and a stochastic program (\textbf{Sto-OPT}).} 
\textcolor{black}{\begin{enumerate}
    \item Qua-E: The forecasting model is trained under MSE and predicts the expected wind power realization. Qua-E corresponds to the forecasting product in Section \ref{Conventional Operation}.
    \item Qua-Q: The forecasting model is trained under pinball loss (asymmetric loss function) and predicts quantile.
    \item Val-L: The value-oriented forecasting approach proposed by \cite{morales2023prescribing} is used as a benchmark, which requires the forecasting model to be linear.
    \item Val-O: The forecasting model trained via OptNet \cite{donti2017task} is used as another benchmark for value-oriented forecasting. 
    \item Per-F: The forecasts perfectly match the wind power realization. Since the forecasting error is inevitable, this is a very ideal benchmark. 
    \item Sto-OPT: We randomly generate 200 wind power scenarios by k-nearest-neighbors \cite{bertsimas2020predictive}. For each sample in the test set, the VPP operator first solves the two-stage stochastic program in \eqref{4} based on the scenarios, and then the real-time problem \eqref{3}. Sto-OPT corresponds to the approach described in Section \ref{Two-stage Stochastic Program based Dispatch}.
\end{enumerate}}

\textcolor{black}{We don't take the approach in Section \ref{Improved Schedule of RES} as the comparison candidate since the performance of this method cannot be better than the approach Sto-OPT in Section \ref{Two-stage Stochastic Program based Dispatch}; see the results in \cite{morales2014electricity}. Therefore, we only compare the proposed approach with the best candidate.} Since a specific type of forecasting model is not the main focus of the work, multi-layer perceptron (MLP) is used as the forecasting model for the proposed approach and benchmarks (Qua-E, Qua-Q, Val-O). The model hyper-parameters are summarized in Table \ref{Table 1}. 

\begin{table}[h]
\caption{Summary of MLP hyper-parameters}\label{Table 1}
\begin{center}
\begin{tabular}{ c  c }
\hline\hline
    Item & Value\\
\hline
    No. of hidden layers & 2\\
    No. of neurons in hidden layer & 256\\
    Dim. of contextual information & 4 \\
    Optimizer & Adam\\
    Learning rate & 1e-3\\
\hline\hline
\end{tabular}
\end{center}
\end{table}


\subsection{Verification Metrics}
We use the root mean square error (RMSE) and the average operating cost to measure the quality and value of the forecasts, respectively, both of which are negative oriented.

\subsubsection{Root Mean Square Error} It is defined as
\begin{equation}
    \text{RMSE} = \sqrt{\frac{1}{D_{test}\cdot T}\sum_{d=1}^{D_{test}}\sum_{\tau=1}^{T}(y_{d,\tau}-\hat{y}_{d,\tau})^2},
\end{equation}
where $D_{test}$ is the number of days in the test set.
\subsubsection{Average Operating Cost} It is defined as,
\begin{equation}\label{18}
\begin{aligned}
    &\frac{1}{D_{test}\cdot T}\sum_{d=1}^{D_{test}}\sum_{\tau=1}^T\bm{\rho}^\top\bm{x}_{d,\tau}^*+\bm{\rho}^{+\top}\bm{z}_{d,\tau}^{+*}-\bm{\rho}^{-\top}\bm{z}_{d,\tau}^{-*},
\end{aligned}
\end{equation}
where for the deterministic dispatch in Section \ref{Conventional Operation}, $
\boldsymbol{x}_{d,\tau}^*$ is the optimal solution of \eqref{2}, i.e., $\boldsymbol{x}_{d,\tau}^*(\hat{y}_{d,\tau})$, with parentheses omitted. For the two-stage stochastic program in Section \ref{Two-stage Stochastic Program based Dispatch}, $
\boldsymbol{x}_{d,\tau}^*$ is the optimal first-stage solution of \eqref{4}. For both the deterministic and two-stage stochastic dispatches, the real-time optimal solutions $\bm{z}_{d,\tau}^{+*},\bm{z}_{d,\tau}^{-*}$ are the ones from \eqref{3}, with parentheses omitted.

\section{Results and Discussion}
\subsection{The Operational Advantage}

In this analysis, the capacity of wind power is scaled to 40 kW, whose capacity is 57\% of that of the maximum demand. In the test set, we use RMSE and average operating cost, calculated by \eqref{18} (which is the summation of the average day-ahead operating cost and average real-time operating cost), as the evaluation metrics for quality and value, respectively. Here, we consider two quality-oriented forecasting approaches, i.e., Qua-E and Qua-Q. The nominal level of the quantile, which is an input to the pinball loss, is set as $\frac{2}{9}$.
Such nominal level is calculated by $\frac{\lambda^*-\nu^{-*}}{\nu^{+*}-\nu^{-*}}$ \cite{pinson2013wind}, where $\lambda^*,\nu^{-*},\nu^{+*}$ are the parameterized dual solutions with the subscript and the parentheses omitted, and are given as an ex-ante “oracle. Specifically, the values are chosen as the marginal generation cost of one of the slow-start generators and flexible resources, i.e., 30 \$/kW, 10 \$/kW, and 100 \$/kW, respectively.
The results of the value- and quality-oriented forecasting are reported in Table \ref{Table 2}. It clearly shows that accurate forecasting does not always benefit the operation. Qua-E achieves a lower RMSE score, which means the forecast shows more correspondence with the realization. Crafted to prioritize value, the proposed approach yields a significant 9.5\% operating cost reduction, compared to Qua-E. Also, Qua-Q is trained under pinball loss, and can partially consider the impact of operating costs and is suitable for the specific case where no operation constraints are considered \cite{pinson2013wind}. Therefore, it results in smaller operating costs, compared to Qua-E. Since it is unsuitable for the more general operation problem discussed in this work, its operating cost is still higher than the proposed one. The results manifest the operational advantage of our approach.

\begin{table}[h]
\caption{RMSE and average operating cost in the test set.}\label{Table 2}
\begin{center}
\resizebox{0.95\columnwidth}{!}{%
\begin{tabular}{ c  c  c  c  c}
\hline\hline
     & \makecell[c]{RMSE\\  $/kW$} & \makecell[c]{Average  \\ operating\\ cost/$\$$} & \makecell[c]{Average \\day-ahead\\ cost/$\$$} & \makecell[c]{Average \\real-time\\ cost/$\$$}\\
\hline
    \makecell[c]{Qua-E} & 7.2 & 39344 & 33191 & 6153\\
    \makecell[c]{Qua-Q} & 9.2 & 36009 & 37221 & -1212\\
        \makecell[c]{Proposed} & 11.1 & 35591 & 38671 & -3080\\
\hline\hline
\end{tabular}}
\end{center}
\end{table}


Fig. \ref{Fig4} displays the 4-day wind power forecast profiles of the value- and quality-oriented forecasting approaches. The real-time problem has a clear influence on value-oriented forecasting. Due to the higher loss for energy shortage than energy surplus, the proposed approach and Qua-Q tend to forecast less wind power production than Qua-E, to avoid the less profitable situation of underproduction, such that the energy shortage is less likely to happen. This point can be further demonstrated by the last two columns of Table \ref{Table 2}, which show that on average, the proposed approach has lower real-time operating cost than Qua-E and Qua-Q. Since the proposed approach tends to forecast less wind power (which has zero marginal cost in the day-ahead problem), the proposed approach has a larger day-ahead operating cost. However, thanks to co-minimizing the overall day-ahead and the real-time costs at the training phase, the proposed approach achieves lower costs for the overall operation.

\begin{figure}[h]
  \centering
  \includegraphics[scale=0.6]{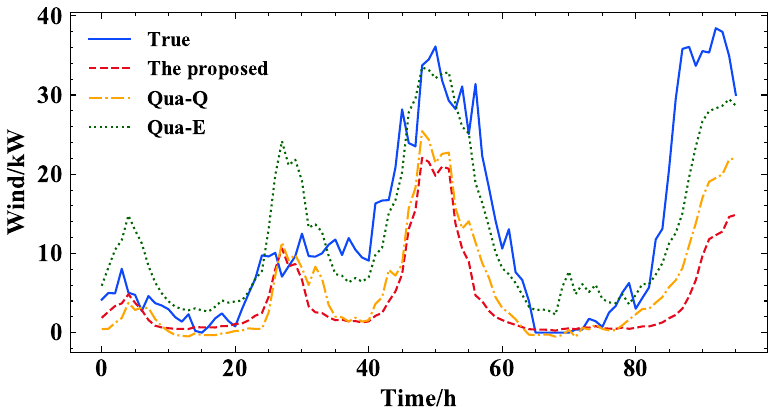}\\
  \caption{4-day wind power forecast profiles issued by the value- and quality-oriented forecasting approaches. Qua-E and Qua-Q stand for quality-oriented forecasts that issue the expectation and $\frac{2}{9}$ quantile of the wind power.}\label{Fig4}
\end{figure}


Furthermore, to demonstrate the advantage of the proposed iterative solution method, we compare it with Val-L \cite{morales2023prescribing} and Val-O \cite{donti2017task}. To accommodate the setting in  \cite{morales2023prescribing} where the ramping constraints are not considered in day-ahead, we omit the ramping constraint in \eqref{2(c)} for this comparison. The average operating cost in the test set is shown in Table \ref{juan}. Additionally, we present the average operating cost resulting from Qua-E. It shows that, due to the integration of more advanced forecasting models (NN v.s. linear model), our approach and Val-O achieve lower average cost, compared to Val-L. 
Additionally, despite employing a more sophisticated model, Qua-E has worse performance compared to Val-L, which employs a linear model. This underscores the necessity of aligning the forecasting model training with the decision value. The training time of neural network-based models is also reported in Table \ref{juan}. As solving the dual problems at the lower level and estimating the forecasting model parameters are iteratively performed, the proposed approach has a longer training time than Qua-E but much shorter than that of the Val-O based on OptNet. Since Val-O needs to calculate the inverse of the Jacobian matrix for the KKT conditions at each epoch, computational cost can be very high.

\begin{table}[h]
\caption{Average operating cost in the test set and training time of neural networks. The ramping constraint \eqref{2(c)} has been dropped from the day-ahead dispatch model in \eqref{2}.}\label{juan}
\begin{center}
\resizebox{0.95\columnwidth}{!}{%
\begin{tabular}{ c c  c  c c}
\hline\hline
& Proposed  & \makecell[c]{Val-L~\cite{morales2023prescribing}} & \makecell[c]{Val-O~\cite{donti2017task}} & \makecell[c]{Qua-E}\\
\hline
    \makecell[c]{Average\\ operating\\ cost} & 35598 \$ &35781 \$ & 35628 \$ & 38667 \$\\
\hline
    \makecell[c]{Training\\ time} & 3 min 50s  &- & 45 min  & 8.26 s \\
\hline\hline
\end{tabular}}
\vspace{-0.5cm}
\end{center}
\end{table}

\textcolor{black}{To make a more nuanced analysis, we identify the sensitive periods when our approach has more obvious improvement compared to Qua-E. The hourly cost reduction of the proposed approach compared to Qua-E is calculated. Also, via the derived loss function \eqref{reloss1}, we identify the factors that affect the hourly cost reduction. The coefficient $\lambda_{d,\tau}^*-\nu_{d,\tau}^*$ is the marginal contribution of the forecasting error $y_{d,\tau}-\Tilde{y}_{d,\tau}$ to the loss function. Therefore, in the following, we analyze the impact of $\lambda_{d,\tau}^*-\nu_{d,\tau}^*$ on the hourly cost reduction. We categorize the hourly $\lambda_{d,\tau}^*-\nu_{d,\tau}^*$ under the proposed forecasts into four intervals whose average values range from low to high. The average hourly cost reduction within the four intervals is calculated and provided in Fig. \ref{reduct}. It shows that the smaller the $\lambda_{d,\tau}^*-\nu_{d,\tau}^*$ is, the more cost reduction our approach can lead to. This is because, in the setting that the real-time marginal cost (the marginal cost of flexible resources addressing energy deficit) is larger than the real-time marginal utility (the marginal utility of flexible resources addressing energy surplus), our approach tends to forecast less wind power to avoid the undesirable energy deficit situation. That is, our approach tends to have $y_{d,\tau}-\Tilde{y}_{d,\tau}\geq 0$. When $\lambda_{d,\tau}^*-\nu_{d,\tau}^*$ is negative, such as the first three points, the term in \eqref{reloss1} is negative in average under the proposed forecasts. This leads to a reduction in the loss and a corresponding reduction in the operating cost, compared to Qua-E. Also, the smaller the $\lambda_{d,\tau}^*-\nu_{d,\tau}^*$ is, the more cost reduction can be expected.}

\begin{figure}[h]
  \centering
  \includegraphics[scale=0.6]{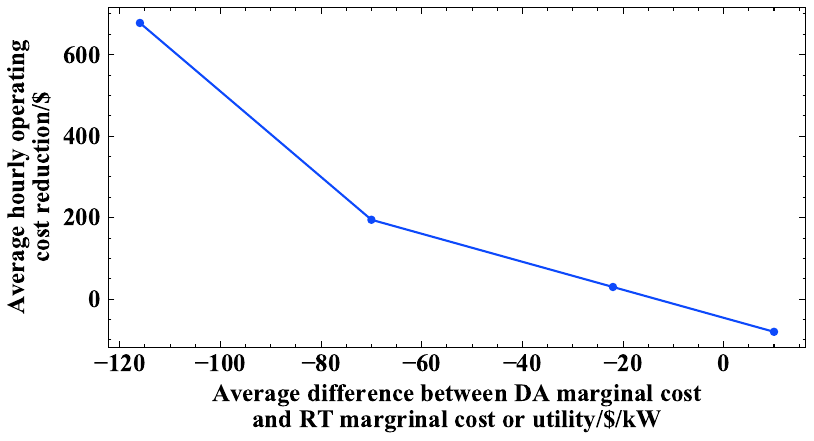}\\
  \caption{Average hourly cost reduction of the proposed approach compared to Qua-E under different levels of difference between DA marginal cost and RT marginal cost/utility.}\label{reduct}
\end{figure}



\subsection{ Sensitivity Analysis: Different Wind Power Penetrations}
This study investigates the performance of value-oriented forecasting under different levels of wind power capacities: 20, 30, and 40 kW. Under these penetration levels, the average operating cost of the proposed approach, Qua-E, and Per-F are summarized in Fig. \ref{Fig5}. 
It shows with the increase in wind power capacity, the average operating cost decreases as wind power, which has zero marginal cost, gradually contributes a larger share toward balancing the load. Also, the trend of cost reduction is more evident in the proposed approach, compared to Qua-E.
Furthermore, under different levels of wind power penetration, the proposed value-oriented forecasting has lower operating costs compared to Qua-E, and has comparable performance with Per-F. The cost reduction is more evident under large wind power capacity. The results show that the benefit of value-oriented forecasting is more significant under the higher penetration of renewable energy resources. 



\begin{figure}[h]
  \centering
  \includegraphics[scale=0.6]{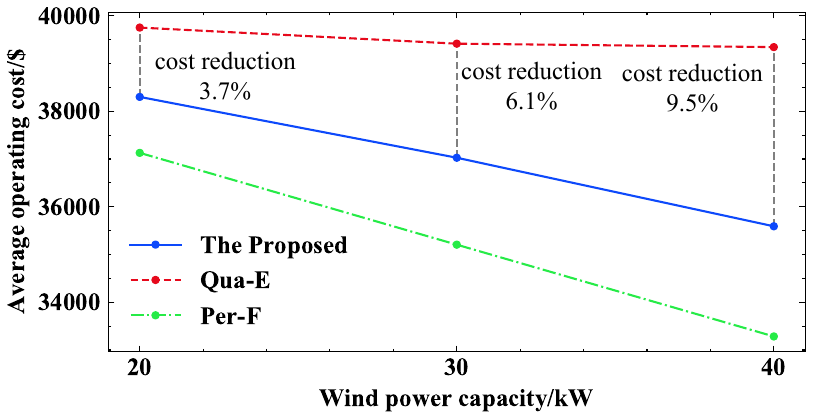}\\
  \caption{Average operating cost under different wind power capacities. Qua-E stands for quality-oriented forecasts predicting expected wind power. Per-F stands for perfect forecasts and therefore results in the lowest operating cost. \vspace{-0.5cm}}
  \label{Fig5}
\end{figure}

\subsection{Computational Complexity Comparison}




\textcolor{black}{Furthermore, to show the proposed approach is more computationally efficient than the two-stage stochastic program Sto-OPT, we compare the computation time of the two approaches in the test set. The scale of the Sto-OPT increases tremendously with the increase of the number of scenarios and the number of flexible resources whose dispatch is determined at the second stage. Therefore, to show the advantage in a larger-scale problem, we consider an operation problem with a greater number of flexible resources. Concretely, there are 10 flexible resources addressing energy shortage (whose capacity is 6 \unit{kW} each and marginal operating cost uniformly ranges from 90 \unit{\$/kW} to 120 \unit{\$/kW}), and 10 flexible resources addressing energy surplus (whose capacity is 6 \unit{kW} each and marginal operating utility uniformly ranges from 10 \unit{\$/kW} to 20 \unit{\$/kW}). The wind power capacity is set as 40 \unit{kW}.}  

\textcolor{black}{We compare the computation time and average operating cost of the two approaches in the test set. The results are in Table \ref{Table 6}. It shows that the proposed approach is much faster than the two-stage stochastic program while achieving a similar average operating cost.}

\begin{table}[h]
\caption{Computation time in the test set and average operating cost of the proposed approach and Sto-OPT.}\label{Table 6}
\begin{center}
\begin{tabular}{ c  c  c  }
\hline\hline
      & 
     \makecell[c]{Computation time/s} & \makecell[c]{Average  operating cost/\$}\\
\hline
    \makecell[c]{Proposed}  &
    3.39  &35846\\
\hline
    \makecell[c]{Sto-OPT}  &
    419  &35837\\
\hline\hline
\end{tabular}
\end{center}
\end{table}

\subsection{Extending to Operational Problems with Binary Variables}\label{ucc}

\textcolor{black}{We show the possibility of applying the proposed approach to the operational problem with binary variables. The day-ahead operation is a unit commitment problem where the binary variables indicate the on/off status of generators. The real-time problem dispatches flexible resources to address the forecast deviation, whose formulation is in \eqref{3}. At the training phase, the day-ahead unit commitment problem is relaxed to a linear program \cite{zhao2023uncertainty}, where the binary variables are relaxed to the continuous ones in the range of 0 and 1. We incorporate the relaxed day-ahead unit commitment and the real-time problem into the proposed approach for training a forecasting model in a value-oriented way. At the operational phase, we apply the issued forecast to serve as the input parameter to the unit commitment with binary variables and use the average overall cost in the test set as the evaluation metric. The mathematical programs of the unit commitment problem and the relaxed one are in Appendix \ref{Appendix B}.} 

\textcolor{black}{Considering different wind power capacities, we compare the average operating cost of the proposed approach and Qua-E. The results are given in Fig. \ref{uccomp}. It shows that the proposed approach can result in a lower average operating cost compared to Qua-E. The cost reduction is more obvious under the large penetration of wind power. Therefore, the proposed approach is effective for the unit commitment problem.}

\begin{figure}[h]
  \centering
  \includegraphics[scale=0.6]{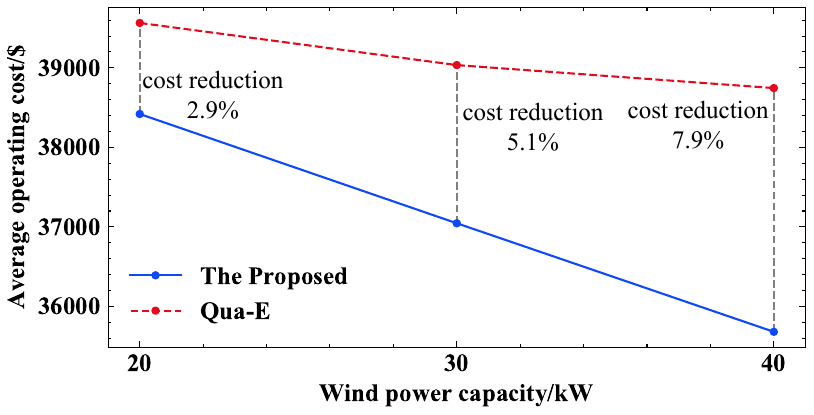}\\
  \caption{Average operating cost of day-ahead unit commitment and real-time energy dispatch problems under different wind power capacities. Qua-E stands for quality-oriented forecasts,  predicting the expected realization.  }
  \label{uccomp}
\end{figure}

\textcolor{black}{Furthermore, we note that the performance can be improved by relaxing the mixed integer program into a linear program in a more precise manner. For instance, \cite{ferber2020mipaal} proposes to generate an exact continuous surrogate for the original discrete optimization problem by adding cutting planes. In this way, the mixed integer program is equivalently transformed into a linear program, which fits the proposed approach. We would like to leave this as the future work.}

\subsection{Adaptability to Operational Change}

\textcolor{black}{In this section, we test our approach under operational changes.} 
Specifically, we consider the case where the parameters in the operational problems are fixed during the training phase. After training, a different set of parameters is used in the operational problems to test the forecasting performance. We provide an example of training the forecasting model with fixed marginal costs in day-ahead and in real-time, but using different real-time marginal costs at the operational phase. 

\textcolor{black}{Two different settings of real-time marginal costs $\bm{\rho}^+=[\rho_1^+,\rho_2^+]^\top$ are considered at the operational phase. In the first setting, the real-time marginal costs $\rho^+_1,\rho^+_2$  are still larger than the real-time marginal utility $\rho^-_1$, but take different values from the training phase. In the second setting, the real-time marginal costs $\rho^+_1,\rho^+_2$  are smaller than the real-time marginal utility $\rho^-_1$. We note that only the real-time marginal costs $\rho^+_1,\rho^+_2$ change at the test phase, and other parameters remain the same as those at the training phase. In Table \ref{setting}, we list the real-time marginal costs $\rho^+_1,\rho^+_2$ of the two settings. In these two settings, we compare the average operating cost in the test set of the proposed approach and Qua-E. The results are in Table \ref{Table 5}.}
\begin{table}[h]
\caption{Real-time marginal cost $\rho^+_1,\rho^+_2$ in the high-cost and low-cost settings.}\label{setting}
\begin{center}
\begin{tabular}{  c  c  c }
\hline\hline
       
     & $\rho^+_1$ & 
     $\rho^+_2$\\
\hline
    High-cost setting& 90 \$/kW&
    100 \$/kW\\
\hline
    Low-cost setting& 5 \$/kW&
    15 \$/kW\\
\hline\hline
\end{tabular}
\end{center}
\end{table}

       

\begin{table}[h]
\caption{Average operating cost of the proposed approach and Qua-E in the high-cost and low-cost settings.}\label{Table 5}
\begin{center}
\begin{tabular}{ c  c  c }
\hline\hline
       
     & Proposed & 
     Qua-E\\
\hline
    High-cost setting& 35511 \$&
    37905 \$\\
\hline
    Low-cost setting& 27578 \$&
    24881 \$\\
\hline\hline
\end{tabular}
\end{center}
\end{table}

\textcolor{black}{It shows that in the high-cost setting, although the values of the marginal costs at the test phase are different from those at the training phase, our approach can adapt to the change, and achieves a lower average cost compared to Qua-E. However, in the low-cost setting, the proposed approach doesn't perform well and has an even larger cost compared to Qua-E. Therefore, in this case, the adaptability of the proposed approach highly depends on how the parameters change. There is no guarantee that the proposed approach can consistently adapt to these changes effectively. A potential solution is considering the cost change at the training phase. We would like to leave this as future work.}

\section{Conclusions}
This paper introduces a novel value-oriented point forecasting methodology tailored for VPP operators who possess renewable energy sources and solve the energy dispatch problems in sequential order. With the forecasts, operators are allowed to solve a computationally efficient deterministic problem during the day-ahead operation, while still attaining low expected operating costs of the day-ahead and real-time stages. To achieve this, we formulate the parameter estimation for the forecasting model as a bilevel program at the training phase. 
An iterative learning strategy is proposed to solve it, where the loss function aligning with the decision value is derived. For the operation problems in the form of LP, such a function exhibits local linearity with regard to forecasts.

Numerical outcomes from the test set reveal that, in comparison to the quality-oriented forecasting methods, the proposed approach yields a larger forecasting error measured by RMSE. However, it manages to attain a reduced average overall operating cost. This advantage becomes more sound, particularly under high levels of wind power integration. The computational efficiency is also demonstrated through a comparison of computation time with a stochastic program. Additionally, we demonstrate the potential of applying the proposed approach to operational problems involving binary variables and adapting to operational changes.

We note that the goal of this study is not to replace the quality-oriented forecasting approach but rather to provide tools for enhanced RES forecasts tailored to specific operational problems to achieve lower operational costs. In the future, we plan to develop value-oriented forecasts that can better adapt to operational changes and applicable to operation problems with network constraints.


\section*{Acknowledgments}

The authors would like to thank Professor Pierre Pinson at Imperial College London for constructive discussion. 

\appendices
\section{Proof of Proposition 1}\label{Appendix A}

Consider the linear programs \eqref{7d}-\eqref{7g} and \eqref{7h}-\eqref{7k}. We derive the dual problems for them, respectively. Taking the problem \eqref{7d}-\eqref{7g} for instance, we rewrite it as
\begin{subequations}\label{22}
\begin{alignat}{2}
 &\mathop{\min}_{\{\bm{x}_{d,\tau}\}_{\tau=1}^T}
 && \quad \sum_{\tau=1}^T\bm{\rho}^\top  \bm{x}_{d,\tau}\label{23(a)}\\
    & \text{s.t.} &&  \quad \bm{A}_D\bm{x}_{d}\leq \bm{b}_D:\bm{\chi}_d\label{22(b)}\\
    &&& \quad \bm{1}^\top  \bm{x}_{d,\tau}+\Tilde{y}_{d,\tau}=l_{d,\tau}:\lambda_{d,\tau},\forall \tau=1,...,T,\label{22(c)}
\end{alignat}
\end{subequations}
where $\bm{\rho},\bm{A}_D,\bm{b}_D$ are the coefficients, $\bm{x}_d=[\bm{x}_{d,\tau}]_{\tau=1,...,T}$, $\bm{\chi}_d=[[\delta_{i,d,\tau}]_{i\in I^D,\tau=1,...,T},[\underline{\bm{\eta}_{d,\tau}},\overline{\bm{\eta}_{d,\tau}}]_{\tau=1,...,T-1}]$. The Lagrangian function of \eqref{22} can be derived as
\begin{equation}\label{23} \bm{\chi}_d^\top(\bm{A}_D\bm{x}_{d}- \bm{b}_D)+\sum_{\tau=1}^T\bm{\rho}^\top \bm{x}_{d,\tau}+\lambda_{d,\tau}(l_{d,\tau}-\bm{1}^\top  \bm{x}_{d,\tau}-\Tilde{y}_{d,\tau}).
\end{equation}

Based on the Lagrangian function \eqref{23}, the objective of the dual problem is
\begin{equation}\label{24}
 -\bm{\chi}_d^\top \bm{b}_D+\sum_{\tau=1}^T-\lambda_{d,\tau}  \Tilde{y}_{d,\tau} +  \lambda_{d,\tau}  l_{d,\tau}. 
\end{equation}

Likewise, we rewrite the linear program \eqref{7h}-\eqref{7k} as
\begin{subequations}\label{25}
\begin{alignat}{2} 
 &\{\mathop{\min}_{\Tilde{\bm{z}}_{d,\tau}}&& \quad \bm{c}_R^{\top}  \Tilde{\bm{z}}_{d,\tau}\label{25(a)}\\
    & \text{s.t.} &&  \quad \bm{A}_R\Tilde{\bm{z}}_{d,\tau}\leq \bm{b}_R:\bm{\iota}_{d,\tau}\label{25(b)}\\
    &&& \quad \bm{d}_R^\top \Tilde{\bm{z}}_{d,\tau}+y_{d,\tau}-\Tilde{y}_{d,\tau}=0:\nu_{d,\tau}\},\forall \tau=1,...,T,\label{25(c)}
\end{alignat}
\end{subequations}
where $\bm{c}_R,\bm{A}_R,\bm{b}_R,\bm{d}_R$ are the coefficients. $\Tilde{\bm{z}}_{d,\tau}=[\Tilde{\bm{z}}_{d,\tau}^+;\Tilde{\bm{z}}_{d,\tau}^-]$, $\bm{c}_R=[\bm{\rho}^+;-\bm{\rho}^-]$ $\bm{d}_R=[\bm{1};-\bm{1}]$ and $\bm{\iota}_{d,\tau}=[[\mu_{i,d,\tau}]_{i \in I^{R+}};[\zeta_{i,d,\tau}]_{i \in I^{R-}}]$.
The objective of dual problem for \eqref{7h}-\eqref{7k} can be derived in a similar way
\begin{equation}\label{26}
 -\nu_{d,\tau}  (y_{d,\tau}-\Tilde{y}_{d,\tau}) -\bm{\iota}_{d,\tau}^\top \bm{b}_R,\forall \tau=1,...,T. 
\end{equation}

The summation of the dual objectives in \eqref{24} and \eqref{26} over $T$ hours consists of the term $\sum_{\tau=1}^T-\lambda_{d,\tau} \Tilde{y}_{d,\tau}-\nu_{d,\tau}(y_{d,\tau}-\Tilde{y}_{d,\tau})$, with the remaining term only regarding the dual variables $\lambda_{d,\tau},\nu_{d,\tau},\bm{\chi}_d,\bm{\iota}_{d,\tau}$. For each sample indexed by $d,\tau$, by plugging the parameterized dual solutions $\lambda_{d,\tau}^*(\Tilde{y}_{d,\tau})$, $\nu_{d,\tau}^*(\Tilde{y}_{d,\tau},y_{d,\tau})$, $\delta_{i,d,\tau}^*(\Tilde{y}_{d,\tau})$,$\underline{\bm{\eta}_{d,\tau}^*}(\Tilde{y}_{d,\tau}),\overline{\bm{\eta}_{d,\tau}}^*(\Tilde{y}_{d,\tau})$, $\mu_{i,d,\tau}^*(\Tilde{y}_{d,\tau},y_{d,\tau})$ and $\zeta_{i,d,\tau}^*(\Tilde{y}_{d,\tau},y_{d,\tau})$ into the objective of the dual problems in \eqref{24},\eqref{26}, the summation of the optimal dual objectives for the sample indexed by $d,\tau$ is
\begin{equation}\label{32}
-\lambda_{d,\tau}^*(\Tilde{y}_{d,\tau}) \Tilde{y}_{d,\tau}-\nu_{d,\tau}^* (\Tilde{y}_{d,\tau},y_{d,\tau})(y_{d,\tau}-\Tilde{y}_{d,\tau})+ \psi^{D,*}_{d,\tau}+\psi^{R,*}_{d,\tau}.
\end{equation}
where $\psi^{D,*}_{d,\tau}$ and $\psi^{R,*}_{d,\tau}$ represent the remaining terms related with the dual solutions. \eqref{32} is equivalent with the term in the curly bracket of \eqref{7a} by the strong duality theory, which ends the proof.

\section{Mathematical Formulation of Unit Commitment Problems in Section \ref{ucc}}\label{Appendix B}

We give the unit commitment problem in day-ahead,
\begin{subequations}\label{uc}
\begin{alignat}{2} 
&\mathop{\min}_{\{\bm{x}_{d,\tau},\bm{u}_{d,\tau}\}_{\tau=1}^T}
  &&\sum_{\tau=1}^{T}\bm{\rho}^\top \bm{x}_{d,\tau}+\bm{\rho}_0^\top\bm{u}_{d,\tau}\label{uca}\\ 
& \text{s.t.}
&& 
\bm{0} \leq \bm{x}_{d,\tau} \leq \overline{\bm{x}} \circ \bm{u}_{d,\tau},\forall \tau=1,...,T \label{uc(b)}
    \\ 
&&& 
 \bm{x}_{d,\tau+1}-\bm{x}_{d,\tau} \leq \bm{r} \circ \bm{u}_{d,\tau+1}, 
\forall \tau=1,...,T-1\label{uc(c)}\\
&&& 
 -\bm{r} \circ \bm{u}_{d,\tau} \leq \bm{x}_{d,\tau+1}-\bm{x}_{d,\tau}, 
\forall \tau=1,...,T-1\label{uc(d)}\\
    &&&  \bm{1}^\top\bm{x}_{d,\tau}+\hat{y}_{d,\tau}=l_{d,\tau},\forall \tau=1,...,T \label{uc(e)}
\end{alignat}
\end{subequations}
where $\circ$ is the elementwise multiplication. $\bm{\rho}$ is the marginal generation cost and $\bm{\rho}_0$ is the start-up cost. $\bm{x}_{d,\tau}$ is the generation schedule at time $\tau$ on day $d$. $\bm{u}_{d,\tau}$ is the binary variable indicating on/off status. $\bm{r}$ is the ramping limit. $\hat{y}_{d,\tau}$ and $l_{d,\tau}$ are the wind power forecast and load demand. By relaxing the binary variable $\bm{u}_{d,\tau}$ to be a continuous variable in the range of 0 and 1, the relaxed unit commitment problem is as follows,
\begin{subequations}\label{relaxeduc}
\begin{alignat}{2} 
&\mathop{\min}_{\{\bm{x}_{d,\tau},\bm{u}_{d,\tau}\}_{\tau=1}^T}
  &&\sum_{\tau=1}^{T}\bm{\rho}^\top \bm{x}_{d,\tau}+\bm{\rho}_0^\top\bm{u}_{d,\tau}\\ 
& \text{s.t.}
&& 
\eqref{uc(b)},\eqref{uc(c)},\eqref{uc(d)},\eqref{uc(e)}
    \\ 
&&& \bm{0} \leq \bm{u}_{d,\tau} \leq \bm{1}
\end{alignat}
\end{subequations}
which is a linear program.

\bibliographystyle{IEEEtran}
\bibliography{IEEEabrv,mylib}

\end{document}